\documentclass[11pt,a4paper]{article}

\usepackage{amsmath,amsfonts,amssymb,amsthm}

\usepackage{graphicx,color}
\usepackage{boxedminipage}
\newtheorem{theorem}{Theorem}

\newtheorem{proposition}{Proposition}

\newtheorem{lemma}{Lemma}

\newcommand{\tw}{{\mathbf{tw}}}
\newcommand{\adj}{\textrm{adj}}
\newcommand{\inc}{\textrm{inc}}
\newcommand{\card}{\textrm{card}}

\DeclareMathOperator{\operatorClassP}{P}
\newcommand{\classP}{\ensuremath{\operatorClassP}}
\DeclareMathOperator{\operatorClassNP}{NP}
\newcommand{\classNP}{\ensuremath{\operatorClassNP}}
\DeclareMathOperator{\operatorClassCoNP}{coNP}
\newcommand{\classCoNP}{\ensuremath{\operatorClassCoNP}}
\DeclareMathOperator{\operatorClassFPT}{FPT}
\newcommand{\classFPT}{\ensuremath{\operatorClassFPT}}
\DeclareMathOperator{\operatorClassW}{W}
\newcommand{\classW}[1]{\ensuremath{\operatorClassW[#1]}}
\DeclareMathOperator{\operatorClassParaNP}{Para-NP}
\newcommand{\classParaNP}{\ensuremath{\operatorClassParaNP}}
\DeclareMathOperator{\operatorClassXP}{XP}
\newcommand{\classXP}{\ensuremath{\operatorClassXP}}

\newcommand{\defparproblem}[4]{
  \vspace{3mm}
\noindent\fbox{
  \begin{minipage}{0.96\textwidth}
  \begin{tabular*}{\textwidth}{@{\extracolsep{\fill}}lr} #1  & {\bf{Parameter:}} #3 \\ \end{tabular*}
  {\bf{Input:}} #2  \\
  {\bf{Question:}} #4
  \end{minipage}
  }
  \vspace{3mm}
}

\begin{document}

\title{Long Circuits and Large  Euler Subgraphs\footnote{Supported by 
  the European Research Council (ERC) via grant Rigorous Theory of Preprocessing, reference 267959.}}

\author{
Fedor V. Fomin\thanks{Department of Informatics, University of Bergen, PB 7803, 5020
Bergen, Norway, \texttt{\{fomin,petr.golovach\}@ii.uib.no}.}
\addtocounter{footnote}{-1}
\and 
Petr A. Golovach\footnotemark
}

\date{}

\maketitle

\begin{abstract}
An undirected graph   is Eulerian if it is connected and all its vertices are of even degree. Similarly, a directed graph is Eulerian, if for each vertex its in-degree is equal to its out-degree. It is well known that  Eulerian graphs can be recognized  in polynomial time while 
the problems of finding a maximum Eulerian subgraph or a maximum induced Eulerian subgraph are NP-hard.
In this paper, we study the parameterized complexity of the following   Euler subgraph problems: 

\begin{itemize}
\item \textsc{Large Euler Subgraph}: For a given graph $G$ and  integer  parameter $k$, does $G$ contain an induced Eulerian subgraph with at least $k$ vertices? 
\item   \textsc{Long Circuit}: For a given graph $G$ and  integer  parameter $k$, does $G$   contain an  Eulerian subgraph with at least $k$ edges? 
\end{itemize} 

Our main algorithmic result is that 
\textsc{Large Euler Subgraph} is fixed parameter tractable (FPT) on undirected graphs. We find this  a bit surprising because the problem of finding an induced Eulerian subgraph with exactly $k$ vertices is known to be W[1]-hard. The complexity of the problem changes drastically on directed graphs. On directed graphs we obtained the following complexity  dichotomy:  
  \textsc{Large Euler Subgraph} is NP-hard for every fixed $k>3$ and is solvable in polynomial time for $k\leq 3$. 
  For \textsc{Long Circuit}, we prove that the problem is FPT on directed and undirected graphs.
\end{abstract}

\section{Introduction} 
One of the oldest theorems in Graph Theory is attributed to Euler, and it says that a (undirected) graph admits an \emph{Euler circuit}, i.e.,  a closed walk visiting every edge exactly once, if and only if a graph is connected and all its vertices are of even degrees. Respectively, a directed graph has a \emph{directed Euler circuit} if and only if a graph is (weakly) connected  and for each vertex, its in-degree is equal to its  out-degree. 
While checking if a given directed or undirected graph is Eulerian  is easily done in polynomial time, the problem of finding $k$ edges (arcs) in a graph to form an Eulerian subgraph is NP-hard. We refer to the book of  Fleischner \cite{Fleischner90} for a thorough study of Eulerian graphs and related topics. 
 
In \cite{CaiY11}, Cai and Yang initiated the study of parameterized complexity of subgraph problems motivated by Eulerian graphs. 
Particularly, they considered the following
parameterized subgraph and induced subgraph problems:

\defparproblem{{\textsc{$k$-Circuit}}}{A (directed) graph $G$ and non-negative integer $k$}{$k$}{Does  $G$ contain a circuit  with $k$ edges (arcs)?}

and

\defparproblem{{\textsc{Euler $k$-Subgraph}}}{A (directed) graph $G$ and non-negative integer $k$}{$k$}{Does  $G$ contain an induced Euler subgraph with $k$ vertices?}

The decision versions of both {\sc $k$-Circuit} and  
{\sc Euler $k$-Subgraph}   are known to be  the \classNP-complete~\cite{CaiY11}.
 Cai and Yang in ~\cite{CaiY11} proved that {\sc $k$-Circuit} on undirected graphs is \classFPT. On the other hand, the authors have  shown in  \cite{FominG12} 
that  {\sc Euler $k$-Subgraph} is \classW{1}-hard.   The variant of the problem  {\sc $(m-k)$-Circuit}, where one asks to remove at most $k$ edges to obtain an Eulerian subgraph  was shown to be     \classFPT \,  by Cygan et al. \cite{CyganMPPS11} 
on directed and undirected graphs. The problem of removing at most $k$ vertices to obtain an induced Eulerian subgraph, namely  {\sc Euler $(n-k)$-Subgraph},  was shown to be  \classW{1}-hard  by Cai and Yang for undirected graphs  \cite{CaiY11} and by Cygan et al. for directed graphs
\cite{CyganMPPS11}. Dorn et al. in \cite{DornMNW13} provided FPT algorithms for the weighted version of Eulerian extension.

\medskip
In this work we extend the set of results on the parameterized complexity of Eulerian subgraph problems by considering the problems of finding an (induced) Eulerian subgraph with \emph{at least} $k$ (vertices) edges. 
 We consider the following problems:

 \defparproblem{{\textsc{Large Euler Subgraph}}}{A (directed) graph $G$ and non-negative integer $k$}{$k$}{Does  $G$ contain  an induced Euler subgraph with  at least $k$ vertices?}
 
and

 \defparproblem{{\textsc{Long Circuit}}}{A (directed) graph $G$ and non-negative integer $k$}{$k$}{Does  $G$ contain a circuit  with at least $k$ edges (arcs)?}

The decision version of  {\sc Long Circuit}  was shown to be  \classNP-complete
by Cygan et al. in~\cite{CyganMPPS11} and it is not difficult to see
that same is true for {\sc Large Euler Subgraph}.   Let us note that by plugging-in these observations into  the framework of   Bodlaender et al.~\cite{BodlaenderDFH09}, it is easy to conclude that on 
undirected graphs both problems have no polynomial kernels unless $\classNP\subseteq\classCoNP/\text{\rm poly}$.

However, the parameterized complexity of these problems appears to be much more interesting. 

\medskip\noindent\textbf{Our results.} We show that 
 {\sc Large Euler Subgraph} behaves differently for directed and undirected cases. For undirected graphs, we prove that 
the problem is \classFPT. We find this result surprising, because very closely related \textsc{Euler $k$-Subgraph} is known to be  \classW{1}-hard \cite{FominG12}.
The proof  is based on a structural result interesting in its own. Roughly speaking, we show that large  treewidth certifies containment of a  large induced Euler subgraph.  For directed graphs,   {\sc Large Euler Subgraph} is \classNP-complete for each  $k\geq 4$, and this bound is tight---the problem is polynomial-time solvable for each $k\leq 3$.
 We also prove that \textsc{Euler $k$-Subgraph} is \classW{1}-hard for directed graphs.
{\sc Long Circuit} is proved to be \classFPT~ for directed and undirected graphs. Our algorithm is based on the results by Gabow and Nie~\cite{GabowN08} about the parameterized complexity of finding long cycles. The known and new results about  Euler subgraph problems are summarized in  Table~\ref{tabl:compl}.

\begin{table}[ht]
\begin{center}
{\small
\begin{tabular}{|c|c|c|}
\hline
& Undirected  & Directed \\
\hline
{\sc $k$-Circuit} & \classFPT~\cite{CaiY11} & \classFPT, Prop.~\ref{prop:exact}\\
\hline
{\sc Euler $k$-Subgraph} &  \classW{1}-hard~\cite{FominG12}&  \classW{1}-hard, Thm~\ref{thm:w-hard}\\
\hline 
{\sc $(m-k)$-Circuit} &  \classFPT~ \cite{CyganMPPS11} &   \classFPT~ \cite{CyganMPPS11} \\
\hline 
 {\sc Euler $(n-k)$-Subgraph} &  \classW{1}-hard~\cite{CaiY11} &    \classW{1}-hard~\cite{CyganMPPS11} \\
\hline
{\sc Long Circuit} & \classFPT, Thm~\ref{thm:fpt-edge} & \classFPT, Thm~\ref{thm:fpt-edge}\\
\hline
{\sc Large Euler Subgraph} & \classFPT, Thm~\ref{thm:fpt-vertex} &  \classNP-complete for any $k\geq 4$,   \\
                                         &                                                           &   Thm~\ref{thm:paranp};  in \classP \, for $k\leq 3$  \\                                 
\hline 
\end{tabular}
\caption{Parameterized complexity of  Euler subgraph problems.}\label{tabl:compl}
}
\end{center}
\end{table}

This paper is organised as follows. Section~\ref{sec:defs} contains basic definitions and preliminaries. In Section~\ref{sec:vertex-undir} we show that 
 {\sc Large Euler Subgraph} is FPT on undirected graphs. In Section~\ref{sec:vertex-dir} we prove that on directed graphs, {\sc Euler $k$-Subgraph}
is W[1]-hard while {\sc Large Euler Subgraph}  is NP-complete for each $k\geq 4$. In Section~\ref{sec:edge} we treat  {\sc Long Curcuit} and show that it is FPT on directed and undirected graphs. 

\section{Basic definitions and preliminaries}\label{sec:defs}

\noindent
{\bf Graphs.}
We consider finite directed and undirected graphs without loops or multiple
edges. The vertex set of a (directed) graph $G$ is denoted by $V(G)$, 
the edge set of an undirected graph and the arc set of a directed graph $G$ is denoted by $E(G)$.
To distinguish edges and arcs, the edge with two end-vertices $u,v$  is denoted by $\{u,v\}$, and we write $(u,v)$ for the corresponding arc.
For a set of vertices $S\subseteq V(G)$,
$G[S]$ denotes the subgraph of $G$ induced by $S$, and by $G-S$ we denote the graph obtained form $G$ by the removal of all the vertices of $S$, i.e., the subgraph of $G$ induced by $V(G)\setminus S$. 
Let $G$ be an undirected graph. 
For a vertex $v$, we denote by $N_G(v)$ its
\emph{(open) neighborhood}, that is, the set of vertices which are
adjacent to $v$. 
The \emph{degree} of a vertex $v$ is denoted by $d_G(v)=|N_G(v)|$, and $\Delta(G)$ is the maximum degree of $G$. 
Let now $G$ be a directed graph.
For a vertex $v\in V(G)$, we say that $u$ is an \emph{in-neighbor} of $v$ if $(u,v)\in E(G)$. The set of all in-neighbors of $v$ is denoted by $N_G^-(v)$. The \emph{in-degree} $d_G^-(v)=|N_G^-(v)|$.
Respectively, $u$ is an  \emph{out-neighbor} of $v$ if $(v,u)\in E(G)$, the set of all out-neighbors of $v$ is denoted by $N_G^+(v)$, and the \emph{out-degree} $d_G^+(v)=|N_G^+(v)|$.

For a (directed) graph $G$, a (directed) \emph{trail} of \emph{length} $k$ is a sequence $v_0,e_1,v_1,e_2,\ldots,e_k,v_k$ of vertices and edges (arcs resp.) of $G$ such that $v_0,\ldots,v_k\in V(G)$, $e_1,\ldots,e_k\in E(G)$,
the edges (arcs resp.) $e_1,\ldots,e_k$ are pairwise distinct, and 
for $i\in\{1,\ldots,k\}$, $e_i=\{v_{i-1},v_i\}$ ($e_i=(v_{i-1},v_i)$ resp.). 
A trail is said to be \emph{closed} if $v_0=v_k$. 
A closed (directed) trail is called a (directed) \emph{circuit}, 
and it is a (directed) \emph{cycle} if all its vertices except $v_0=v_k$ are distinct.
Clearly, any cycle is a subgraph of $G$, and it is said that $C$ is an \emph{induced cycle} of $G$ if $C=G[V(C)]$. A (directed) path is a trail such that all its vertices  are distinct. 
For a (directed) walk (trail, path resp.) $v_0,e_1,v_1,e_2,\ldots,e_k,v_k$, $v_0$ and $v_k$ are its \emph{end-vertices}, and $v_1,\ldots,v_{k-1}$ are its \emph{internal} vertices. For a 
a (directed) walk (trail, path resp.) with end-vertices $u$ and $v$, we say that it is an \emph{$(u,v)$-walk} (\emph{trail}, \emph{path} resp.).
We omit the word ``directed'' if it does not create a confusion. Also we write a trail as a sequence of its vertices $v_0,\ldots,v_k$.

A connected (directed) graph $G$ is an \emph{Euler} (or \emph{Eulerian}) graph if it has a (directed) circuit that contains all edges (arcs resp.) of $G$. By the famous result by Euler (see, e.g., \cite{Fleischner90}), a connected graph $G$ is an Euler graph if and only if all its vertices have even degrees. Respectively, a connected directed graph $G$ is an Euler directed graph if and only if for each vertex $v\in V(G)$, $d_G^-(v)=d_G^+(v)$.

\medskip
\noindent
{\bf Ramsey numbers.} The \emph{Ramsey number} $R(r,s)$ is the minimal integer $n$ such that any graph on $n$ vertices has either a clique of size $r$ or an independent set of size $s$.
By the famous paper by Erd{\"o}s and Szekeres~\cite{ErdosS35}, 
$R(r,s)\leq \binom{r+s-2}{r-1}$.

\medskip
\noindent
{\bf Parameterized Complexity.}
Parameterized complexity is a two dimensional framework
for studying the computational complexity of a problem. One dimension is the input size
$n$ and another one is a parameter $k$. It is said that a problem is \emph{fixed parameter tractable} (or \classFPT), if it can be solved in time $f(k)\cdot n^{O(1)}$ for some function $f$, 
and
it is said that a problem is in  \classXP, if it can be solved in time $O(n^{f(k)})$ for some function $f$.
One of the basic assumptions of the Parameterized Complexity theory is the conjecture that the complexity class $\classW{1}\neq \classFPT$, and it is unlikely that a \classW{1}-hard problem could be solved in \classFPT-time. A problem is \emph{\classParaNP-hard}(\emph{complete}) if it is \classNP-hard (complete) for some fixed value of the parameter $k$.
Clearly, a \classParaNP-hard problem is not in \classXP~ unless \classP$=$\classNP.
We refer to the books of Downey and Fellows~\cite{DowneyF99}, 
Flum and Grohe~\cite{FlumG06}, and   Niedermeier~\cite{Niedermeierbook06} for  detailed introductions  to parameterized complexity. 

\medskip
\noindent
{\bf Treewidth.} 
A \emph{tree decomposition} of a graph $G$ is a pair $(X,T)$ where $T$
is a tree and $X=\{X_{i} \mid i\in V(T)\}$ is a collection of subsets (called {\em bags})
of $V(G)$ such that: 
\begin{enumerate}
\item $\bigcup_{i \in V(T)} X_{i} = V(G)$, 
\item for each edge $\{x,y\} \in E(G)$, $x,y\in X_i$ for some  $i\in V(T)$, and 
\item for each $x\in V(G)$ the set $\{ i \mid x \in X_{i} \}$ induces a connected subtree of $T$.
\end{enumerate}
The \emph{width} of a tree decomposition $(\{ X_{i} \mid i \in V(T) \},T)$ is $\max_{i \in V(T)}\,\{|X_{i}| - 1\}$. The \emph{treewidth} of a graph $G$ (denoted as $\tw(G)$) is the minimum width over all tree decompositions of $G$. 

%
%
\medskip
Recall that our aim is to prove that  {\sc Long Circuit} is \classFPT~ for directed and undirected graphs and 
{\sc Large Euler Subgraph} is \classFPT~ for undirected graphs. Hence, we should argue that these problems are \classNP-hard.
Notice that for exact variants of these problems, i.e.,  {\sc $k$-Circuit} and  
{\sc Euler $k$-Subgraph} for undirected graphs, the \classNP-hardness was proved by Cai and Yang~\cite{CaiY11}, and 
for  {\sc Long Circuit} (for directed and undirected cases), it was shown 
by Cygan et al. in~\cite{CyganMPPS11}. 

\begin{proposition}\label{prop:NPc}
{\sc Large Euler Subgraph} is \classNP-complete for undirected graphs when $k$ is   a part of the input.
\end{proposition}
 
\begin{proof}
Let $G$ be a $n$-vertex cubic graph. Denote by $G'$ the graph obtained by subdividing each edge of $G$. It is straightforward to see that $G'$ has an induced Euler subgraph with at least $2n$ vertices if and only if $G$ is Hamiltonian.
As the {\sc Hamiltonian Cycle} is known to be \classNP-complete for 
cubic planar graphs~\cite{GareyJT76}, the result follows.
\end{proof}

\medskip
Finally, we observe that while we obtain \classFPT~ results, it is unlikely that these problems have polynomial kernels. Let $G$ be a (directed) graph with $t$ connected components $G_1,\ldots,G_t$, and let $k$ be a positive integer. Notice that $G$ has a circuit with at least $k$ edges (arcs resp.) if and only if $G_i$ has a circuit with at least $k$ edges (arcs resp.) for some $i\in\{1,\ldots,t\}$. Also $G$ has an induced Euler subgraph with at least $k$ vertices if and only if $G_i$ has an induced Euler subgraph with at least $k$ vertices for some $i\in\{1,\ldots,t\}$. By the results of Bodlaender et al.~\cite{BodlaenderDFH09}, this observation together with Proposition~\ref{prop:NPc} and the NP-hardness of  {\sc Long Circuit}~\cite{CyganMPPS11} immediately imply the following proposition.

\begin{proposition}\label{prop:kernel}
 {\sc Long Circuit} for directed and undirected graphs and {\sc Large Euler Subgraph} for undirected graphs have no polynomial kernels unless $\classNP\subseteq\classCoNP/\text{\rm poly}$.
\end{proposition}

\section{Large Euler subgraphs}\label{sec:EulerSubgr}

\subsection{Large Euler subgraphs for undirected graphs}\label{sec:vertex-undir}
In this section we show that {\sc Large Euler Subgraph} is \classFPT~ for undirected graphs. Using Ramsey arguments, we prove that if a graph $G$ has sufficiently large treewidth, then $G$ has an induced Euler subgraph on at least $k$ vertices. Then if the input graph has large treewidth, we have a YES-answer. Otherwise, we use the fact that {\sc Large Euler Subgraph} can be solved in \classFPT~ time for graphs of bounded treewidth. All considered here graphs are undirected.

For a given positive integer $k$, we define the function $f(\ell)$ for integers $\ell\geq 2$ recursively as follows:
\begin{itemize}
\item $f(2)=R(k,k-1)+1$,
\item $f(\ell)=(k-1)(2(\ell-1)(f(\lfloor\frac{\ell}{2}\rfloor+1)-1)+1)+1$ for $\ell>2$.
\end{itemize}  

We need the following two lemmas.

\begin{lemma}\label{lem:connect}
Let $G$ be a graph, and suppose that $s,t$ are distinct vertices joined by at least $f(\ell)$ internally vertex-disjoint paths of length at most $\ell$ in $G$ for some $\ell\geq 2$. Then $G$ has an induced Euler subgraph on at least $k$ vertices.  
\end{lemma}

\begin{proof}
Consider the minimum value of $\ell$ such that $G$ has $f(\ell)$ internally vertex disjoint $(s,t)$-paths.
We have at least $r=f(\ell)-1$ such paths $P_1,\ldots,P_r$ that are distinct from the path  $s,t$. 
We assume that each path $P_i$ has no chords that either join two internal vertices or an internal vertex and one of the end-vertices, i.e., each internal vertex is adjacent in $G[V(P_i)]$ only to its two neighbors in $P_i$. Otherwise, we can replace $P_i$ by a shorter path with all vertices in $V(P_i)$ distinct from the path $s,t$.  
We consider two cases.

\medskip
\noindent
{\bf Case 1. $\mathbf{ \ell=2}$.} The paths $P_1,\ldots,P_r$ are of length two and therefore have exactly one internal vertex. Assume that $u_1,\ldots,u_{r}$ are internal vertices of these paths. Because $r=f(2)-1=R(k,k-1)$, the graph $G[\{u_1,\dots,u_{r}\}]$ either has a clique $K$ of size $k$ or an independent set $I$ of size at least $k-1$. 
Suppose that $G$ has a clique $K$. If $k$ is odd, then $G[K]$ is an induced Euler subgraph on $k$ vertices. If $k$ is even, then $G[K\cup\{s\}]$ is an induced Euler subgraph on $k+1$ vertices. 
Assume now that that $I\subseteq\{u_1,\ldots,u_{r-1}\}$ is an independent set of size $k-1$. Let $v\in I$. 
If $\{s,t\}\in E(G)$ and $k$ is even or $\{s,t\}\notin E(G)$ and $k$ is odd, then $G[I\cup\{s,t\}]$ is an induced Euler subgraph on $k+1$ vertices.
Else if $\{s,t\}\notin E(G)$ and $k$ is even or $\{s,t\}\in E(G)$ and $k$ is odd, then $G[I\cup\{s,t\}\setminus \{v\}]$ is an induced Euler subgraph on $k$ vertices.

\medskip
\noindent
{\bf Case 2. $\mathbf{ \ell\geq3}$.} 
We say that paths $P_i$ and $P_j$ are \emph{adjacent} if they have adjacent internal vertices.
Let $p=f(\lfloor \ell/2\rfloor+1)$.
Suppose that there is an internal vertex $v$ of one of the paths $P_1,\ldots,P_r$ that is adjacent to at least 
$2p-1$ internal vertices of some other distinct $2p-1$ paths.
Then there are $p=f(\lfloor \ell/2\rfloor+1)$ paths $P_{i_1},\ldots,P_{i_p}$ that have respective internal vertices $v_1,\ldots,v_p$ such that i) $v$ is adjacent to $v_1,\ldots,v_p$ and
ii) either each $v_j$ is at distance at most $\lfloor \ell/2\rfloor$ from $s$ in $P_{i_j}$ for all $j\in\{1,\ldots,p\}$ or  
each $v_j$ is at distance at most $\lfloor \ell/2\rfloor$ from $t$ in $P_{i_j}$ for all $j\in\{1,\ldots,p\}$. But then either the vertices $s,v$ or $v,t$ are joined by at least 
$f(\lfloor \ell/2\rfloor+1)$ internally vertex-disjoint paths of length at most $\lfloor \ell/2\rfloor+1<\ell$. This contradicts our choice of $\ell$. 
Hence, for each $i\in \{1,\ldots,r\}$, any internal vertex of $P_i$ has adjacent internal vertices in at most $2p-2$ other paths, and $P_i$ is adjacent to at most $2(\ell-1)(p-1)$ other paths. 
As $r=(k-1)(2(\ell-1)(p-1)+1)$, there are $k-1$ distinct paths $P_{i_1},\ldots,P_{i_{k-1}}$ that are pairwise non-adjacent, i.e., they have no adjacent internal vertices.

Let $H=G[V(P_{i_1})\cup\ldots\cup V(P_{i_{k-1}})]$ and $H'=G[V(P_{i_1})\cup\ldots\cup V(P_{i_{k-2}})]$.
Notice that by our choice of the paths, $H=P_{i_1}\cup\ldots\cup P_{i_{k-1}}$ and $H'=P_{i_1}\cup\ldots\cup P_{i_{k-2}}$ if $\{s,t\}\notin E(G)$, and 
$P_{i_1}\cup\ldots\cup P_{i_{k-1}}$ ($P_{i_1}\cup\ldots\cup P_{i_{k-2}}$ resp.) can be obtained from $H$ ($H'$ resp.) by the removal of $\{s,t\}$ if $s,t$ are adjacent.
If $\{s,t\}\in E(G)$ and $k$ is even or $\{s,t\}\notin E(G)$ and $k$ is odd, then $H$ is an induced Euler subgraph on at least $k+1$ vertices.
Else if $\{s,t\}\notin E(G)$ and $k$ is even or $\{s,t\}\in E(G)$ and $k$ is odd, then $H'$ is an induced Euler subgraph on at least $k$ vertices.
\end{proof}

For $k\geq 4$, let
$$\Delta_k=1+\frac{(f(3k-8)-1)((f(3k-8)-2)^{3(k-3)}-1)}{f(3k-8)-3}.$$

\begin{lemma}\label{lem:deg}
For $k\geq 4$,
any 2-connected graph $G$ with $\Delta(G)>\Delta_k$ 
has an induced Euler subgraph on at least $k$ vertices. 
\end{lemma}

\begin{proof}
Let $G$ be a 2-connected graph and let $u$ be a vertex of $G$ with $d_G(u)=\Delta(G)$. As $G$ is 2-connected, $G'=G-u$ is connected. Let $v$ be an arbitrary vertex of $N_G(u)$. Denote by $T$ a tree of shortest paths from $v$ to all other vertices of $N_G(u)$ in $G'$. 

\medskip
\noindent
{\bf Claim A. } {\it 
If there is a $(v,w)$-path $P$ of length at least $3(k-3)+1$ in $T$ for some $w\in N_G(u)$, then $G$ has an induced Euler subgraph on at least $k$ vertices. 
} 

\begin{proof}[Proof of Claim~A]
Denote by $v_0,\ldots,v_r$ the vertices of $N_G(u)$ in $P$ in the path order, $v_0=v$ and $v_r=w$. 
Let $Q_1$ be the union of $(v_0,v_1),(v_3,v_4),\ldots,(v_{\lfloor r/3\rfloor},v_{\lfloor r/3\rfloor+1})$-subpaths of $P$, let $Q_2$ be the union of $(v_1,v_2),(v_4,v_5),\ldots,(v_{\lfloor r/3\rfloor+1},v_{\lfloor r/3\rfloor+2})$-subpaths of $P$, and let $Q_3$
be the union of $(v_2,v_3),(v_5,v_6),\ldots,(v_{\lfloor r/3\rfloor-1},v_{\lfloor r/3\rfloor})$-subpaths of $P$.
Notice that some subpaths can be empty depending whether $r$ modulo 3 is 0, 1 or 2. Observe that $Q_1,Q_2,Q_3$ are edge-disjoint induced subgraphs of $G$.
Since $Q_1\cup Q_2\cup Q_3=P$, there is $Q_i$ for $i\in\{1,2,3\}$ with at least $k-2$ edges. Then $Q_i$ has at least $k-1$ vertices. Let $H=G[V(Q_i)\cup \{u\}]$. By the definition of $Q_i$, $H$ is a union of induced cycles with one common vertex $u$ such that for different cycles $C_1,C_2$ in the union, 
$V(C_1)\cap V(C_2)=\{u\}$ and $\{x,y\}\notin E(G)$ whenever $x\in V(C_1)\setminus \{u\}$ and $y\in V(C_2)\setminus\{u\}$. Hence, $H$ is an Euler graph with at least $k$ vertices. 
\end{proof}

From now we assume that all $(v,w)$-paths in $T$ have length at most $3(k-3)$ for $w\in N_G(u)$.

\medskip
\noindent
{\bf Claim B. } {\it 
If there is a vertex $w\in V(T)$ with $d_T(w)\geq f(3(k-3)+1)$, then $G$ has an induced Euler subgraph on at least $k$ vertices. 
} 
 
\begin{proof}[Proof of Claim~B]
Recall that $T$ is a tree of shortest paths from $v$. We assume that $T$ is rooted in $v$. Then the roof defines the parent-child relation on $T$. Let $x_0$ be a parent of $w$ (if exists) and let $x_1,\ldots,x_r$ be the children of $w$.  If $w$ has no parent, then $w=v$ and $r\geq f(3(k-3)+1)$. Otherwise, $r\geq f(3(k-3)+1)-1$. Let $y_0=v$. For each $i\in\{1,\ldots,r\}$, let $y_i\in V(T)\cap N_G(u)$ be a descendant of $x_i$ in $T$. Denote by $P_i$ the unique $(w,y_i)$-path in $T$ for $i\in\{0,\ldots,r\}$.
As all $(v,w)$-paths in $T$ have length at most $3(k-3)$ for $w\in N_G(u)$, the paths 
$P_0,\ldots,P_r$ have length at most $3(k-3)$. Notice that these paths have no common vertices except $w$.
Observe now that $y_0,\ldots,y_r$ are adjacent to $u$ in $G$. Therefore, we have at least  $f(3(k-3)+1)$ 
internally vertex-disjoint $(u,w)$-paths in $G$. By Lemma~\ref{lem:connect}, it implies that $G$ 
has an induced Euler subgraph on at least $k$ vertices. 
\end{proof}

To complete the proof of the lemma, it remains to observe that if $\Delta(T)<f(3(k-3)+1)$ and all $(v,w)$-paths in $T$ have length at most $3(k-3)$ for $w\in N_G(u)$, then $$d_G(u)\leq |V(T)|\leq 1+\frac{(f(3k-8)-1)((f(3k-8)-2)^{3(k-3)}-1)}{f(3k-8)-3}=\Delta_k.$$
\end{proof}

Kosowski et al. \cite{KosowskiLNS12} obtained the following bound for treewidth.

\begin{theorem}[\cite{KosowskiLNS12}]\label{thm:bod}
Let $G$ be a graph without induced cycles with at least $k\geq 3$ vertices and let $\Delta(G)\geq 1$. Then $\tw(G)\leq k(\Delta(G)-1)+2$. 
\end{theorem}

This theorem together with Lemma~\ref{lem:deg} immediately imply the next lemma.

\begin{lemma}\label{lem:tw-bound}
Let $G$ be a graph and let $k\geq 4$. If 
$\tw(G)> k(\Delta_k-1)+2$, 
then $G$ has an induced Euler subgraph on at least $k$ vertices. 
\end{lemma}

\begin{proof}
Suppose that $\tw(G)> k(\Delta_k-1)+2$. Then $G$ has a 2-connected component $G'$ with $\tw(G\rq{})> k(\Delta_k-1)+2$. If $\Delta(G')>\Delta_k$, then $G'$ has  an induced Euler subgraph on at least $k$ vertices by Lemma~\ref{lem:deg}.  Otherwise, by Theorem~\ref{thm:bod}, $G'$ has an induced cycle on at least $k$ vertices, i.e., an induced Euler subgraph. 
\end{proof}

Now we show that {\sc Large Euler Subgraph} is \classFPT~ for graphs of bounded treewidth. 

\begin{lemma}\label{lem:cmsol}
For any positive integer $t$, {\sc Large Euler Subgraph} can be solved in linear time for graphs of treewidth at most $t$.
\end{lemma}

\begin{proof}
Recall that the syntax of the \emph{monadic second-order logic} (MSOL) of graphs includes logical connectives $\vee$, $\land$, $\neg$, 
variables for vertices, 
edges, sets of vertices and edges, 
and quantifiers $\forall$, $\exists$ that can be applied to these variables. 
Besides the standard relations $=$, $\in$, $\subseteq$, 
the syntax includes the relation
${\adj}(u,v)$ for two vertex variables, which expresses  whether two vertex $u$ and $v$ are adjacent, and for a vertex variable $v$ and an edge variable $e$, we have the relation $\inc(v,e)$ which expresses that $v$ is incident with $e$. The \emph{counting monadic first-order logic} (CMSOL) is an extension of MSOL with the additional predicate $\card_{p,q}(X)$ which expresses whether the cardinality of a set $X$ is $p$ modulo $q$.

By the celebrated Courcelle's theorem~\cite{Courcelle92}, any problem that can be expressed in CMSOL can be solved in linear time for graphs of bounded treewidth. Moreover, this result holds also for optimization  problems.  

Observe that to solve {\sc Large Euler Subgraph} 
for a graph $G$, it is sufficient to find a subset of vertices $U$ of maximum size such that $U$ induces an Euler graph.
Clearly, $U$ induces an Euler graph if and only if i) $G[U]$ is connected and ii) each vertex of $G[G]$ has even degree. The both properties can be expressed in CMSOL. The standard way to express connectivity is to notice that $G[U]$ is connected if and only if for any $X\subset U$, $X\neq\emptyset$ and $X\neq U$, there is an edge $\{x,y\}\in E(G)$ such that $x\in X$ and $y\in U\setminus X$. Then we have to express the property that for any $u\in U$, $d_{G[U]}(u)=|N_{G[U]}(u)|$ is even. To do it, it is sufficient to observe that $X=N_{G[U]}(u)$ if and only if $X\subseteq U$ such that i) for any $v\in X$, $\{u,v\}\in E(G)$, and ii) for any $v\in U$ such that $\{u,v\}\in E(G)$, $v\in X$. Since we can express in CMSOL whether $|N_{G[U]}(u)|$ is even, the claim follows. 
\end{proof}

Now we can prove the main result of this section. 

\begin{theorem}\label{thm:fpt-vertex}
For any positive integer $k$, {\sc Large Euler Subgraph} can be solved in linear time for undirected graphs.
\end{theorem}

\begin{proof}
Clearly, we can assume that $k\geq 3$, as any Euler graph has at least three vertices. If $k=3$, then we can find any shortest cycle in the input graph $G$. It is straightforward to see that if $G$ has no cycles, then we have no Euler subgraph, and any induced cycle is an induced Euler subgraph on at least three vertices. Hence, it can be assumed that $k\geq 4$. We check in linear time whether $\tw(G)> k(\Delta_k-1)+2$ using the Bodlaender's algorithm~\cite{Bodlaender96}. If it is so, we solve our problem using Lemma~\ref{lem:cmsol}. Otherwise, by Lemma~\ref{lem:tw-bound}, we conclude that $G$ has an induced Euler subgraph on at least $k$ vertices and return a YES-answer.  
\end{proof}

Notice, that the proof of Theorem~\ref{thm:fpt-vertex} is not constructive. 
In the conclusion of this section we sketch the algorithm that produces an  induced Euler subgraph on at least $k\geq4$ vertices if it exists. 

First, for each $\ell\geq 2$, we can test existence of two vertices $s,t$ such that the input graph $G$ has at least $f(\ell)$ internally vertex-disjoint $(s,t)$-paths of length at most $\ell$ in \classFPT~time with the parameter $\ell$ using the color coding technique~\cite{GolovachT11}. If we find such a structure for $\ell\leq 3k-8$, we find an induced Euler subgraph with at least $k$ vertices that is either a clique or a union of $(s,t)$-paths as it is explained in the proof of Lemma~\ref{lem:connect}.

Otherwise, we find all 2-connected components. If there is a 2-connected component $G'$ with a vertex $u$ with $d_{G'}(u)>\Delta_k$, then we find an induced Euler subgraph with at least $k$ vertices that is a union of cycles with the common vertex $u$ using the arguments form the proof of Lemma~\ref{lem:deg}.

If all 2-connected components have bounded maximum degrees, we use the algorithm of Kosowski et al.~\cite{KosowskiLNS12} that in polynomial time either finds an induced cycle on at least $k$ vertices or constructs a tree decomposition of width at most $k(\Delta_k-1)+2$. In the fist case we have an induced Euler subgraph on at least $k$ vertices. In the second case the treewidth is bounded, and 
{\sc Large Euler Subgraph}  is solved by a dynamic programming algorithm instead of applying Lemma~\ref{lem:cmsol}.

\subsection{Large Euler subgraphs for directed graphs}\label{sec:vertex-dir}
In this section we show that 
{\sc Euler $k$-Subgraph} and {\sc Large Euler Subgraph} are hard for directed graphs. 

First, we consider {\sc Euler $k$-Subgraph}. It is straightforward to see that this problem is in \classXP, since we can check for every subset of $k$ vertices, whether it induces an Euler subgraph. We prove that this problem cannot be solved in $\classFPT$~ time unless $\classFPT=\classW{1}$. 

\begin{theorem}\label{thm:w-hard}
The   {\sc Euler $k$-Subgraph} is \classW{1}-hard for directed graphs.
\end{theorem}

\begin{proof}
We reduce  the {\sc Multicolored $k$-Clique} problem that is known to be \classW{1}-hard~\cite{FellowsHRV09}:

 \defparproblem{{\textsc{Multicolored $k$-Clique}}}{A $k$-partite graph $G=(V_1\cup\dots\cup V_k,E)$, where $V_1,\ldots,V_k$  are sets of the $k$-partition}{$k$}{Does  $G$ has a clique with $k$ vertices?}

Let $G=(V_1\cup\dots\cup V_k,E)$ be an instance of {\sc Multicolored $k$-Clique}. We construct the directed graph $H$ as follows.
\begin{itemize}
\item Construct the copies of $V_1,\ldots,V_k$.
\item For $1\leq i<j\leq k$ and for each  $u\in V_i$ and $v\in V_j$, if $\{u,v\}\notin E(G)$, then construct an arc $(u,v)$ for the copies of $u$ and $v$ in $H$.
\item For each $i\in\{1,\ldots,k\}$, construct two vertices $x_i,y_i$, then join $x_i$ by arcs with all vertices of $V_i$ and join every vertex of $V_i$ with $y_i$ by an arc.
\item Construct arcs $(y_1,x_2),(y_2,x_3),\ldots,(y_k,x_1)$.
\end{itemize}
We set $k'=3k$. 

We claim that $G$ has a clique with $k$ vertices if and only if $H$ has an induced Euler subgraph with at least $k'$ vertices.

Let $K=\{v_1,\ldots,v_k\}$ be a clique in $G$, $v_i\in V_i$ for $i\in\{1,\ldots,k\}$. Observe that 
$v_1,\ldots,v_k$
are pairwise non-adjacent in $H$. Hence, the set of vertices $\{x_1,v_1,y_1,\ldots,x_k,v_k,y_k\}$ induces a cycle in $H$. Hence, we have an induced Euler subgraph with at least $k'$ vertices.

Suppose now that $H$ has an induced Euler subgraph $C$ with at least $k'$ vertices. Observe that every directed cycle in $G$ contains the arc $(y_k,x_1)$, because if we delete this arc, we obtain a directed acyclic graph. Since any Euler directed graph is a union of arc-disjoint directed cycles (see, e.g.,  \cite{Fleischner90}), $C$ is an induced directed cycle. Moreover, for each $i\in\{1,\ldots,k\}$, $C$ contains at most one vertex of $V_i$. Indeed, assume that two vertices $u,v$ of $C$ are in the same set $V_i$. Then the $(u,v)$-paths and the $(v,u)$-path in $C$ should contain $(y_k,x_1)$, but this is impossible. Because $C$ has $3k$ vertices, we conclude that $C$ contains exactly one vertex $v_i$ from each $V_i$ for $i\in\{1,\ldots,k\}$, and $x_i,y_i\in V(C)$ for $i\in\{1,\ldots,k\}$. Then $C=H[\{x_1,v_1,y_1,\ldots,x_k,v_k,y_k\}]$. Since $C$ is an induced cycle, $v_1,\ldots,v_k$ are pairwise non-adjacent in $H$. Then $\{v_1,\ldots,v_k\}$ is a clique in $G$. 
\end{proof}

For {\sc Large Euler Subgraph} for directed graphs, we prove that this problem is \classParaNP-complete.

\begin{theorem}\label{thm:paranp}
For any $k\geq 4$, {\sc Large Euler Subgraph} is \classNP-complete for directed graphs.
\end{theorem}

\begin{proof}
We reduce to the {\sc 3-Satisfiability} problem. It is known~\cite{ECCC-TR03-049} that this problem is \classNP-complete even if each variable is used in exactly 4 clauses: 2 times in positive and 2 times in negations. 

\begin{figure}[ht]
\centering\scalebox{0.7}{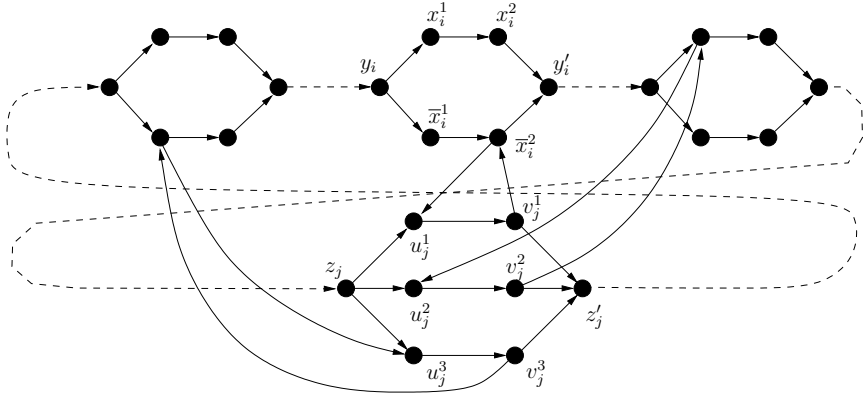}
\caption{Construction of $G$. 
\label{fig:npc}}
\end{figure}

Suppose that Boolean variables $x_1,\ldots,x_n$ and clauses $C_1,\ldots,C_m$ compose an instance of  {\sc 3-Satisfiability} such that each variable is used in exactly 4 clauses: 2 times in positive and 2 times in negations. Let $\phi=C_1\land\ldots\land C_m$. Without loss of generality we assume that $k\leq 4(n+m)$.  We construct the directed graph $G$ as follows (see Fig.~\ref{fig:npc}).
\begin{itemize}
\item For each $i\in\{1,\ldots,n\}$, construct vertices $y_i,y_i',x_i^1,x_i^2,\overline{x}_i^1,\overline{x}_i^2$ and arcs $(y_i,x_i^1), (x_i^1,x_i^2),(x_i^2,y_i')$,  $(y_i,\overline{x}_i^1), (\overline{x}_i^1,\overline{x}_i^2),(\overline{x}_i^2,y_i')$.
\item For each $i\in\{2,\ldots,n\}$, construct $(y_{i-1}',y_i)$.
\item For each $j\in\{1,\ldots,m\}$, construct vertices $z_j,z_j',u_j^1,u_j^2,u_j^3,v_j^1,v_j^2,v_j^3$ and arcs $(z_j,u_j^1)$, $(u_j^1,v_j^1)$, $(v_j^1,z_j')$, $(z_j,u_j^2)$, $(u_j^2,v_j^2)$, $(v_j^2,z_j')$, 
$(z_j,u_j^3)$, $(u_j^3,v_j^3)$, $(v_j^3,z_j')$.
\item For each $j\in\{2,\ldots,m\}$, construct $(z_{j-1}',z_j)$.
\item Construct $(y_n',z_1),(z_m',y_1)$.
\item For each $j\in \{1,\ldots,m\}$, let $C_j=w_1\vee w_2\vee w_3$. For $h\in\{1,2,3\}$,
\begin{itemize}
\item if $w_h=x_i$ for some $i\in\{1,\dots,n\}$ and $w_h$ is the fist occurrence of the literal $x_i$ in $\phi$, then construct $(x_i^1,u_j^h),(v_j^h,x_i^1)$;
\item if $w_h=x_i$ for some $i\in\{1,\dots,n\}$ and $w_h$ is the second occurrence of the literal $x_i$ in $\phi$, then construct $(x_i^2,u_j^h),(v_j^h,x_i^2)$;
\item if $w_h=\overline{x}_i$ for some $i\in\{1,\dots,n\}$ and $w_h$ is the fist occurrence of the literal $\overline{x}_i$ in $\phi$, then construct $(\overline{x}_i^1,u_j^h),(v_j^h,\overline{x}_i^1)$;
\item if $w_h=\overline{x}_i$ for some $i\in\{1,\dots,n\}$ and $w_h$ is the second occurrence of the literal $\overline{x}_i$ in $\phi$, then construct $(\overline{x}_i^2,u_j^h),(v_j^h,\overline{x}_i^2)$.
\end{itemize}
\end{itemize}

We claim that $\phi$ can be satisfied if and only if $G$ has an induced Euler subgraph with at least $k$ vertices.

For $i\in\{1,\ldots,n\}$, let $P_i=\{y_i,x_i^1,x_i^2,y_i'\}$ and $\overline{P}_i=\{y_i,\overline{x}_i^1,\overline{x}_i^2,y_i'\}$. For $j\in\{1,\ldots,n\}$ and $h\in\{1,2,3\}$, $Q_j^h=\{z_j,u_j^h,v_j^h,z_j'\}$. Let also 
$Z_j=\{z_1,\ldots,z_m\}$ and $Z_j'=\{z_1',\ldots,z_m'\}$.

Suppose that the variables $x_1,\ldots,x_n$ have values such that $\phi$ is satisfied. We construct the set of vertices $U$ as follows:
\begin{itemize}
\item for $i\in\{1,\ldots,n\}$, if $x_i=true$, then the vertices of the set $P_i$ are included in $U$, and $\overline{P}_i$ is included otherwise;
\item for $j\in\{1,\ldots,m\}$, if $C_j=w_1\vee w_2\vee w_3$, then we choose a literal $w_h=true$ in $C_j$ and include $Q_j^h$ in $U$.
\end{itemize} 
Observe that $U$ induces a cycle in $G$ and $|U|=4(n+m)\geq k$. Then $G[U]$ is an induced Euler subgraph on at least $k$ vertices.

Assume now that a set $U\subseteq V(G)$ induces an Euler graph $H=G[U]$ and $|U|\geq k$. 

Since $k\geq 4$,  $U\cap (Z\cup Z')\neq \emptyset$ because $G-(Z\cup Z')$ has no cycles except vertex-disjoint triangles.  Suppose that some $z_j\in U$. 
If $j=1$, then $y_n'\in U$, and if $j>1$, then $z_{j-1}'\in U$.  Also exactly one of the vertices $u_j^1,u_j^2,u_j^3$ is in $U$. Let $u_j^h\in U$. Then $v_j^h\in U$. Further, either $z_j'\in U$ or 
some vertex $x_i^p\in U$ or $\overline{x}_i^p\in U$. But if $x_i^p\in U$ or $\overline{x}_i^p\in U$, $d_{H}^-u_j^h=2>1=d_H^+(u_j^h)$.  Hence, $z_j'\in U$. If $j=m$, then $y_1\in U$, and if $j<m$, then $z_{j+1}\in U$.   
Suppose now that some $z_j'\in U$. 
If $j=m$, then $y_1\in U$, and if $j<m$, then $z_{j+1}\in U$.  Also exactly one of the vertices $v_j^1,v_j^2,v_j^3$ is in $U$. Let $v_j^h\in U$. 
Then $u_j^h\in U$. Further, either $z_j\in U$ or 
some vertex $x_i^p\in U$ or $\overline{x}_i^p\in U$. But if $x_i^p\in U$ or $\overline{x}_i^p\in U$, $d_{H}^-v_j^h=1<2=d_H^+(v_j^h)$.  Hence, $z_j\in U$. If $j>1$, then $z_{j-1}'\in U$.  We have that for each $j\in \{1,\ldots,m\}$, exactly one $Q_h\subseteq U$ for $h\in\{1,2,3\}$, and for $h'\in\{1,2,3\}\setminus\{h\}$, $u_j^{h'},v_j^{h'}\notin U$. 
Also we have that $y_1\in U$.

Now we consequently consider $i=1,\ldots,n$. We already know that $y_1\in U$. Assume inductively that $y_i\in U$. Then exactly one of the vertices $x_i^1,\overline{x}_i^1$ is in $U$. 
Assume that  $x_i^1$ is in $U$ as another case is symmetric. Then either $x_i^2$ or some vertex $u_j^h$ should be in $U$. But if  $u_j^h\in U$, then because $z_j\in U$, $d_H^-(u_j^h)=2>d_H^+(u_j^h)$ and this is impossible for an Euler graph. Then $x_j^2\in U$. By the same arguments we show that $y_i'\in U$. If $i<n$, then $y_{i+1}\in U$, and we can proceed with our inductive arguments. We conclude that  for each $i\in \{1,\ldots,n\}$, either $P_i\subseteq U$ or $\overline{P}_i\subseteq U$, and if $P_i\subseteq U$ ( $\overline{P}_i\subseteq U$ resp.), then $\overline{x}_i^1,\overline{x}_i^2\notin U$ ($x_i^1,x_i^2\notin U$ resp.). Moreover, if $P_i\subseteq U$ ($\overline{P}_i\subseteq U$ resp.)  and $Q_j^h\subseteq U$, then $H$ has no arcs between the vertices of these two sets.  

We define the truth assignment for the variables $x_1,\ldots,x_n$ as follows:
for each $i\in \{1,\ldots,n\}$, if $P_i\subseteq U$, then $x_i=false$, and $x_i=true$ otherwise. We claim that $\phi$ is satisfied by this assignment.
To show it, consider a clause $C_j=w_1\vee w_2\vee w_3$. We know that there is $h\in\{1,2,3\}$ such that $Q_j^h\subseteq U$. Assume that $w_h=x_i$ (the case $w_h=\overline{x}_i$ is symmetric). Then $Q_i^h$ is joined by arcs in $G$ with $P_i$.  It follows that $P_i$ was not included in $U$, i.e., $\overline{P}_i\subseteq U$ and $x_i=true$. Hence, $w_h=true$. 
\end{proof}

We proved that  {\sc Large Euler Subgraph} is \classNP-complete for directed graphs for $k\geq 4$. In the conclusion of this section we observe that the bound $k\geq 4$ is tight unless \classP$=$\classNP.

\begin{proposition}\label{prop:poly}
 {\sc Large Euler Subgraph} can be solved in polynomial time for $k\leq 3$.
\end{proposition}

\begin{proof}
For $k=1$, the problem is trivial. If $k=2$, then any shortest cycle $C$ in a directed graph $G$ is an induced Euler subgraph of $G$ with at least two vertices, and  $G$ has no induced Euler subgraph if $G$ is a directed acyclic graph. Hence, it remains to consider $k=3$.

Suppose that $H$ is an induced Euler subgraph of a directed graph $G$, and $H$ has at least three vertices. Denote by $H'$ the graph obtained from $H$ by the deletion of all pairs of opposite arcs, i.e., for each pair of vertices $x,y$ such that $(x,y),(y,x)\in E(H)$, we delete $(x,y),(y,x)$. Clearly, for any $v\in V(H')$, $d_{H'}^-(v)=d_{H'}^+(v)$. If $H'$ is empty, then because $|V(H)|\geq 3$, there are three distinct vertices $x,y,z\in V(H)$ such that $(x,y),(y,x),(y,z),(z,y)\in E(H)$ and either $(x,z),(z,x)\in E(H)$ or $(x,z),(z,x)\notin E(H)$. Then $G[\{x,y,z\}]$ is an Euler subgraph of $G$. If $H'$ is non-empty, then $H'$ has a shortest cycle $C$. Because $H'$ has no cycles of length two,  $G[V(C)]$ is an induced Euler subgraph with at least three vertices.  

We conclude that {\sc Large Euler Subgraph} can be solved for $k=3$ as follows.
If $G$ has three distinct vertices $x,y,z$ such that $(x,y),(y,x),(y,z),(z,y)\in E(G)$ and either $(x,z),(z,x)\in E(G)$ or $(x,z),(z,x)\notin E(G)$,  then $G[\{x,y,z\}]$ is an Euler subgraph of $G$.
Otherwise, let $G'$ be the graph obtained from $G$ by the deletion of all pairs of opposite arcs. We find a shortest cycle $C$ in $G'$, and we have that $G[V(C)]$ is an induced Euler subgraph with at least three vertices. Finally, if $G'$ is a directed acyclic graph, then we return a NO-answer. 
\end{proof}

\section{Long circuits}\label{sec:edge}
In this section we show that {\sc Long Circuit} problem is \classFPT~ for directed and undirected graphs.


\defparproblem{{\textsc{At least $k$ and at most $k'$-Circuit}}}{A (directed) graph $G$ and non-negative integers $k,k'$, $k\leq k'$}{$k'$}{Does  $G$ contain a circuit with at least $k$ and at most $k'$ edges (arcs)?}

Clearly, we can solve this problem in \classFPT~ time for undirected graphs applying the algorithm by Cai and Yang~\cite{CaiY11}  for  
{\sc $r$-Circuit} for each $r\in\{k,\ldots,k'\}$. For directed graphs, 
we can use the same approach based on the color coding technique introduced by Alon, Yuster and Zwick~\cite{AlonYZ95}. For completeness, we sketch the proof here.

\begin{lemma}\label{lem:atmost}
The {\sc At least $k$ and at most $k'$-Circuit} problem can be solved 
in $2^{O(k')}\cdot nm$ expected
time and in $2^{O(k')}\cdot nm\log n$ worst-case time for (directed) graphs with $n$ vertices and $m$ edges (arcs).
\end{lemma} 

\begin{proof}
As the algorithms for directed and undirected graphs are basically the same, we consider here the directed case. For simplicity, we solve the decision problem, but the algorithm can be easily modified to obtain a circuit of prescribed size if it exists.

Let $G$ be a directed graph with $n$ vertices and $m$ arcs.

First, we describe the randomized algorithm. We color the arcs of $G$ by $k'$ colors $1,\ldots,k'$ uniformly at random independently from each other. 
Now we are looking for a \emph{colorful} circuit in $G$ that has at least $k$ arcs, i.e., for a circuit such that all arcs are colored by distinct colors.

To do it, we apply the dynamic programming across subsets. We choose an initial vertex $u$ and try to construct a circuit that includes $u$.  
For a set of colors $X\subseteq\{1,\ldots,k'\}$, denote by $U(X)$ the set of vertices $v\in V(G)$ such that  
there is a $(u,v)$-trail with $|X|$ edges colored by distinct colors from $X$. It is straightforward to see that $U(\emptyset)=\{u\}$. For $X\neq\emptyset$, $v\in U(X)$ if and only if $v$ has an in-neighbor $w\in N_G^-(v)$ such that $(w,v)$ is colored by a color $c\in X$ and $w\in U(X\setminus\{c\})$.
We consequently construct the sets $U(X)$ for $X$ with $1,2,\ldots, k'$ elements. We stop and return a YES-answer if $u\in U(X)$ for some $X$ of size at least $k$.   
Notice that the sets $U(X)$ can be constructed in time $O(k'2^{k'}\cdot m)$. Since we try all possibilities to select $u$, the running time is $O(k'2^{k'}\cdot mn)$.

Now we observe that for any positive number $p<1$, there is a constant $c_p$ such that after running our randomized algorithm $c_p2^{O(k')}$ times, we either get a YES-answer or can claim that with probability $p$ 
$G$ has no directed circuit with at least $k$ and at most $k'$ arcs. 

This algorithm can be derandomized by the  technique proposed by Alon, Yuster and Zwick~\cite{AlonYZ95}. To do it, we replace random colorings by a family of at most $2^{O(k')}\log n$ hash functions that can be constructed in time $2^{O(k')}\cdot m\log n$. 
\end{proof}

If we set $k'=k$, then Lemma~\ref{lem:atmost} immediately implies the following proposition. Notice that it was proved in~\cite{CaiY11} for undirected graphs.

\begin{proposition}\label{prop:exact}
The {\sc $k$-Circuit} problem can be solved 
in $2^{O(k)}\cdot nm$ expected
time and in $2^{O(k)}\cdot nm\log n$ worst-case time for (directed) graphs with $n$ vertices and $m$ edges (arcs).
\end{proposition}   

Gabow and Nie in~\cite{GabowN08} considered the {\sc Long Cycle} problem:

\defparproblem{{\textsc{Long Cycle}}}{A (directed) graph $G$ and a positive integer $k$}{$k$}{Does  $G$ contain a cycle with at least $k$ edges (arcs)?}

In particular, they proved the following theorem.

\begin{theorem}[\cite{GabowN08}]\label{thm:gabow-dir}
The {\sc Long Cycle} problem can be solved 
in $2^{O(k\log k)}\cdot nm$ expected
time and in $2^{O(k\log k)}\cdot nm\log n$ worst-case time for directed graphs with $n$ vertices and $m$ arcs.
\end{theorem}

Let us recall that a fundamental cycle in undirected graph is formed from a spanning tree and a nontree edge.
For  undirected graphs, it is slightly more convenient to use the structural result by Gabow and Nie. 

\begin{theorem}[\cite{GabowN08}]\label{thm:gabow-undir}
In a connected undirected graph having a cycle with $k$ edges, either
every depth-first search tree has a fundamental cycle with at least $k$ edges or some cycle with at least $k$ edges has at most $2k-4$
edges.
\end{theorem}

We need the following observation.

\begin{lemma}\label{lem:bound-length}
Let $G$ be a (directed) graph without cycles of length at least $k$. If $G$ has a circuit with at least $k$ edges (arcs resp.), then $G$ has a circuit with at least $k$ and at most $2k-2$ edges (arcs resp.).
\end{lemma}

\begin{proof}
Let $C$ be a circuit in   $G$. It is well-known (see, e.g., \cite{Fleischner90}) that $C$ is a union of edge-disjoint (arc-disjoint) cycles $C_1,\ldots,C_r$. Moreover, it can be assumed that for any $i\in\{1,\ldots,r\}$, the circuit $C_1\cup\ldots\cup C_i$ is connected. 
Suppose now that $C$ is a circuit with at least $k$ edges (arcs resp.) that has minimum length. Then the circuit  $C'=C_1\cup\ldots\cup C_{r-1}$ has at most $k-1$ edges (arcs resp.).
Since $G$ has no cycles of length at least $k$, $C_r$ has at most $k-1$ edges (arcs resp.). Thus $C$ has at most $2k-2$ edges (arcs resp.). 
\end{proof}

Now we are ready to prove the main result of this section. 

\begin{theorem}\label{thm:fpt-edge}
The {\sc Long Circuit} problem can be solved 
in $2^{O(k\log k)}\cdot nm$ expected
time and in $2^{O(k\log k)}\cdot nm\log n$ worst-case time for directed graphs with $n$ vertices and $m$ arcs, and
in $2^{O(k)}\cdot nm$ expected
time and in $2^{O(k)}\cdot nm\log n$ worst-case time for undirected graphs with $n$ vertices and $m$ edges.
\end{theorem}

\begin{proof}
First, we consider directed graphs. Let $G$ be a directed graph.  By Theorem~\ref{thm:gabow-dir}, we can check whether  $G$ has a cycle with at least $k$ arcs. If we find such a cycle $C$, then $C$ is a circuit with at least $k$ arcs, and we have a YES-answer. Otherwise, we conclude that each cycle in $G$ is of  length at most $k-1$. Then by Lemma~\ref{lem:bound-length}, if $G$ has a circuit with at least $k$ arcs, then it has a circuit with at least $k$ and at most $2k-2$ arcs. We find such a circuit (if it exist) by solving {\sc At least $k$ and at most $k'$-Circuit} for $k'=2k-2$ by making use of  Lemma~\ref{lem:atmost}. Combining the running times, we have that   {\sc Long Circuit} can be solved 
in $2^{O(k\log k)}\cdot nm$ expected time and in $2^{O(k\log k)}\cdot nm\log n$ worst-case time. 

For the undirected case, we assume that the input graph $G$ is connected, as otherwise we can solve the problem for each component. We choose a vertex $v$ arbitrarily and perform  the depth-first search from $v$. In this way we find the fundamental cycles for the dfs-tree rooted in $v$,  and check whether there is a fundamental cycle of length at least $k$.   If we have such a cycle, then it is 
a circuit with at least $k$ edges, and we have a YES-answer. Otherwise, 
by Theorem~\ref{thm:gabow-undir}, either $G$ has no cycles of length at least $k$ or $G$ has a cycle with at least $k$ and at most $2k-4$ edges. If $G$ has no cycles with at least $k$ edges, then by Lemma~\ref{lem:bound-length}, if $G$ has a circuit with at least $k$ edges, it contains  a circuit with at least $k$ and at most $2k-2$ edges. We conclude that if the constructed fundamental cycles have lengths at most $k-1$, then $G$ either has a circuit with at least $k$ and at most $2k-2$ edges or has no circuit with at least $k$ edges. We check whether $G$ has a circuit with at least $k$ and at most $2k-2$ edges  by solving {\sc At least $k$ and at most $k'$-Circuit} for $k'=2k-2$ using Lemma~\ref{lem:atmost}.
Since the depth-first search runs in linear time, we have that  on undirected graphs  {\sc Long Circuit} can be solved 
in $2^{O(k)}\cdot nm$ expected time and in $2^{O(k)}\cdot nm\log n$ worst-case time. 
\end{proof}


\begin{thebibliography}{10}

\bibitem{AlonYZ95}
{\sc N.~Alon, R.~Yuster, and U.~Zwick}, {\em Color-coding}, J. ACM, 42 (1995),
  pp.~844--856.

\bibitem{ECCC-TR03-049}
{\sc P.~Berman, M.~Karpinski, and A.~D. Scott}, {\em Approximation hardness of
  short symmetric instances of {MAX-3SAT}}, Electronic Colloquium on
  Computational Complexity (ECCC),  (2003).

\bibitem{Bodlaender96}
{\sc H.~L. Bodlaender}, {\em A linear-time algorithm for finding
  tree-decompositions of small treewidth}, SIAM J. Comput., 25 (1996),
  pp.~1305--1317.

\bibitem{BodlaenderDFH09}
{\sc H.~L. Bodlaender, R.~G. Downey, M.~R. Fellows, and D.~Hermelin}, {\em On
  problems without polynomial kernels}, J. Comput. Syst. Sci., 75 (2009),
  pp.~423--434.

\bibitem{CaiY11}
{\sc L.~Cai and B.~Yang}, {\em Parameterized complexity of even/odd subgraph
  problems}, J. Discrete Algorithms, 9 (2011), pp.~231--240.

\bibitem{Courcelle92}
{\sc B.~Courcelle}, {\em The monadic second-order logic of graphs III:
  tree-decompositions, minor and complexity issues}, ITA, 26 (1992),
  pp.~257--286.

\bibitem{CyganMPPS11}
{\sc M.~Cygan, D.~Marx, M.~Pilipczuk, M.~Pilipczuk, and I.~Schlotter}, {\em
  Parameterized complexity of {E}ulerian deletion problems}, in WG 2011,
  vol.~6986 of Lecture Notes Comp. Sci., Springer, 2011, pp.~131--142.

\bibitem{DornMNW13}
{\sc F.~Dorn, H.~Moser, R.~Niedermeier, and M.~Weller}, {\em Efficient
  algorithms for eulerian extension and rural postman}, SIAM J. Discrete
  Math.,, 27 (2013), pp.~75--94.

\bibitem{DowneyF99}
{\sc R.~G. Downey and M.~R. Fellows}, {\em Parameterized complexity},
  Monographs in Computer Science, Springer-Verlag, New York, 1999.

\bibitem{ErdosS35}
{\sc P.~Erd{\"o}s and G.~Szekeres}, {\em A combinatorial problem in geometry},
  Compositio Math., 2 (1935), pp.~463--470.

\bibitem{FellowsHRV09}
{\sc M.~R. Fellows, D.~Hermelin, F.~A. Rosamond, and S.~Vialette}, {\em On the
  parameterized complexity of multiple-interval graph problems}, Theor. Comput.
  Sci., 410 (2009), pp.~53--61.

\bibitem{Fleischner90}
{\sc H.~Fleischner}, {\em {Eulerian} Graphs and Related Topics, Part 1, Volume
  1}, Annals of Discrete Mathematics 45, Amsterdam, 1990.

\bibitem{FlumG06}
{\sc J.~Flum and M.~Grohe}, {\em Parameterized complexity theory}, Texts in
  Theoretical Computer Science. An EATCS Series, Springer-Verlag, Berlin, 2006.

\bibitem{FominG12}
{\sc F.~V. Fomin and P.~A. Golovach}, {\em Parameterized complexity of
  connected even/odd subgraph problems}, in STACS 2012, vol.~14 of LIPIcs,
  Schloss Dagstuhl - Leibniz-Zentrum fuer Informatik, 2012, pp.~432--440.

\bibitem{GabowN08}
{\sc H.~N. Gabow and S.~Nie}, {\em Finding a long directed cycle}, ACM
  Transactions on Algorithms, 4 (2008).

\bibitem{GareyJT76}
{\sc M.~R. Garey, D.~S. Johnson, and R.~E. Tarjan}, {\em The planar hamiltonian
  circuit problem is np-complete}, SIAM J. Comput., 5 (1976), pp.~704--714.

\bibitem{GolovachT11}
{\sc P.~A. Golovach and D.~M. Thilikos}, {\em Paths of bounded length and their
  cuts: Parameterized complexity and algorithms}, Discrete Optimization, 8
  (2011), pp.~72--86.

\bibitem{KosowskiLNS12}
{\sc A.~Kosowski, B.~Li, N.~Nisse, and K.~Suchan}, {\em k-chordal graphs: From
  cops and robber to compact routing via treewidth}, in ICALP 2012, vol.~7392
  of Lecture Notes in Computer Science, Springer, 2012, pp.~610--622.

\bibitem{Niedermeierbook06}
{\sc R.~Niedermeier}, {\em Invitation to fixed-parameter algorithms}, vol.~31
  of Oxford Lecture Series in Mathematics and its Applications, Oxford
  University Press, Oxford, 2006.

\end{thebibliography}

\end{document}